\documentclass[a4paper]{article}

\usepackage{amsmath,amssymb,amsthm}
\usepackage{graphicx}
\usepackage{a4wide}
\usepackage{cite}
\usepackage{url}
\usepackage{listings}
\usepackage{fancyhdr}
\usepackage{color}

\newtheorem{theorem}{Theorem}
\newtheorem{definition}[theorem]{Definition}
\newtheorem{lemma}[theorem]{Lemma}
\newtheorem{remark}[theorem]{Remark}

\newcommand\blfootnote[1]{%
	\begingroup
	\renewcommand\thefootnote{}\footnote{#1}%
	\addtocounter{footnote}{-1}%
	\endgroup
}

\begin{document}

\pagestyle{fancy}

\fancyhf{} 
\fancyhead[L]{Modeling Blood Alcohol Concentration}
\fancyfoot[C]{\thepage}

\title{Modeling Blood Alcohol Concentration Using Fractional 
Differential Equations Based on the $\psi$-Caputo Derivative\blfootnote{This is 
a preprint of a paper whose final and definite form is published Open Access 
in \emph{Math. Meth. Appl. Sci.} at [http://doi.org/10.1002/mma.10002].}} 

\author{Om Kalthoum Wanassi$^{1,2,}$\thanks{Corresponding author: {\tt om.wanassi@ua.pt}}\\
{\tt om.wanassi@ua.pt}
\and Delfim F. M. Torres$^{1,3}$\\
{\tt delfim@ua.pt}}

\date{$^1$\text{Center for Research and Development in Mathematics and Applications (CIDMA),}\\
Department of Mathematics, University of Aveiro, 3810-193 Aveiro, Portugal\\[0.3cm]
$^2$LR18ES17, Faculty of Sciences of Monastir, University of Monastir,\\ 
5019 Monastir, Tunisia\\[0.3cm]
$^3$Research Center in Exact Sciences (CICE), Faculty of Sciences and Technology (FCT),
University of Cape Verde (Uni-CV), Praia 7943-010, Cape Verde}

\maketitle


\begin{abstract}
We propose a novel dynamical model for blood alcohol concentration 
that incorporates $\psi$-Caputo fractional derivatives. Using the 
generalized Laplace transform technique, we successfully derive
an analytic solution for both the alcohol concentration in the stomach 
and the alcohol concentration in the blood of an individual. 
These analytical formulas provide us a straightforward numerical scheme, 
which demonstrates the efficacy of the $\psi$-Caputo derivative operator 
in achieving a better fit to real experimental data on blood alcohol 
levels available in the literature. In comparison to existing classical 
and fractional models found in the literature, our model outperforms 
them significantly. Indeed, by employing a simple yet non-standard kernel 
function $\psi(t)$, we are able to reduce the error by more than half, 
resulting in an impressive gain improvement of 59 percent.

\medskip


\medskip

\end{abstract}

\maketitle


\section{Introduction}
\label{sec:01}

Alcohol, a toxic and psychoactive substance known for its addictive properties, 
has become deeply integrated into many societies. Alcoholic beverages have become 
a commonplace element of social interactions for a significant portion of the population. 
This is especially evident in social environments that carry considerable visibility 
and societal influence, where alcohol often accompanies social gatherings. 
Regrettably, the detrimental health and social consequences caused or exacerbated 
by alcohol are frequently overlooked or downplayed. In reality, alcohol consumption 
is responsible for a staggering three million deaths worldwide each year, while millions 
more suffer from disabilities and poor health as a result. The harmful use of alcohol 
accounts for approximately 7.1\% of the global burden of disease among males and 2.2\% 
among females. Shockingly, alcohol stands as the primary risk factor for premature 
mortality and disability among individuals aged 15 to 49, comprising 10\% of all deaths 
within this age group. Moreover, disadvantaged populations, particularly those who 
are vulnerable, experience disproportionately higher rates of alcohol-related 
deaths and hospitalizations \cite{Global,Health}.

Over the past several decades, fractional calculus has captured the attention 
of researchers across diverse fields of science and engineering \cite{Ref3.02}. 
Fractional differential equations, in particular, have emerged as a common tool 
in various scientific and engineering disciplines \cite{MR3917292,Ref3.04}. 
These equations find application in fields such as signal processing, 
control theory, diffusion, thermodynamics, biophysics, blood flow phenomena, 
rheology, electrodynamics, electrochemistry, electromagnetism, continuum mechanics, 
statistical mechanics, and dynamical systems 
\cite{Agrawal_et_al,Baleanu_2007,Benghorbal,Wanassi-Delfim,Singh}.

The ability to model blood alcohol content over time is not only of interest 
to medical professionals but also holds significant value in comparing metabolic 
capabilities. Mathematical techniques provide a means to not only model blood 
concentration but also to analyze the metabolic processes of various endogenous compounds, 
such as blood glucose levels or administered medications.
Recent research \cite{Ricardo_et_al, Qureshi_et_al} has focused on investigating models 
for blood alcohol concentration. A comparative analysis involving three types of fractional 
derivatives -- Caputo, Atangana-Baleanu-Caputo (ABC), and Caputo-Fabrizio (CF) derivatives --
demonstrates that the Caputo and ABC operators are better suited for numerical simulations 
using real data when compared to classical models employing standard integer-order derivatives 
\cite{Qureshi_et_al}. Here, we improve the best results of \cite{Ricardo_et_al,Qureshi_et_al}
by using $\psi$-Caputo fractional derivatives, which allows us to reduce
the total square error by more than half, resulting in an impressive gain of 59 percent.
The selection of this generalized-Caputo operator in our work 
is grounded in several merits that make it a suitable choice for our study
due to its versatility, suitability for modeling fractional order systems, 
and its ability to capture complex dynamics. Indeed, the generalized $\psi$-Caputo 
operator is a versatile fractional derivative 
that provides a unified framework for handling a wide range of real-world problems. 
It has been successfully applied in various scientific disciplines, including physics, 
engineering, and mathematical modeling \cite{MR4492865}.
Its versatility allows us to tackle complex phenomena and systems in a unified manner,
being particularly well-suited for modeling systems with memory effects and non-local behavior. 
As we shall see, it allows us to capture the fractional order dynamics of blood alcohol accurately.
By using the generalized-Caputo operator, we enhance the fidelity of our model, 
enabling us to better capture the long-range dependencies and memory effects 
that play a crucial role in modeling blood alcohol concentration. 
This improved modeling precision lead us to more accurate predictions 
and a better understanding of the underlying processes. The choice of the 
generalized $\psi$-Caputo operator represents one of the novel aspects of our manuscript 
compared to previous articles in the field. By using this operator, we contribute 
to the expanding body of knowledge on fractional calculus and its applications: 
our choice allows us to offer a fresh perspective and advance considerably 
the state of the art.

The paper is organized as follows. In Section~\ref{sec:02}, 
we recall the notions and results from $\psi$-fractional
calculus needed in the sequel. Our contributions
are then given in Section~\ref{sec:03}: we introduce the new $\psi$-Caputo 
fractional model and obtain an explicit formula for the exact solution of the problem
(Section~\ref{sec:03:1}); we show how available models and results from the literature
can be obtained as particular cases (Sections~\ref{sec:3.2.1} and \ref{sec:3.2.2});
and we show the accuracy and efficiency of our new model 
by significantly improving the available results 
in the literature (Section~\ref{sec:3.2.3}).
We end with Section~\ref{sec:05} of conclusion.


\section{Preliminaries}
\label{sec:02}

Originally, the $\psi$-Caputo fractional calculus was introduced by Osler in 1970 
\cite{MR0260942}, being now part of the classical fractional calculus: see  \cite[Section~18.2]{MR1347689} 
and \cite[Section~2.5]{MR2218073}. Recently, Almeida made a small variation 
on the Riemann--Liouville operators to get the Caputo versions, 
and popularized the $\psi$-terminology \cite{Almeida1}.
Here we recall necessary notions from this calculus
and two lemmas that will be useful in the proof
of out theoretical result.

\begin{definition}[The $\psi$-Riemann--Liouville fractional integral \cite{MR1347689,MR2218073}]
\label{def:01}
Let $\alpha >0$, $f: [a,b] \longrightarrow \mathbb{R}$ be 
integrable and $\psi \in \mathcal{C}^{1}([a,b])$ an increasing function 
such that $\psi'(t)\neq 0$, for all $t \in [a,b]$. The $\psi$-Riemann--Liouville 
fractional integral of $f$ of order $\alpha$ is defined as follows:
\begin{equation}
\label{IF-RL}
I_{a^+}^{\alpha,\psi} f(t)=\frac{1}{\Gamma(\alpha)}
\int_{a}^{t} \psi'(s)(\psi(t)-\psi(s))^{\alpha-1}f(s) ds, 
\quad \alpha>0,
\end{equation}
where $\Gamma(\alpha)$ is the Gamma function. 
\end{definition}

Note that for $\psi(t)=t$ and $\psi(t)=\ln(t)$, equation (\ref{IF-RL}) 
is reduced to  the Riemann--Liouville and Hadamard fractional integrals,
respectively.

\begin{definition}[The $\psi$-Riemann--Liouville fractional derivative \cite{MR1347689,MR2218073}]
\label{def:02}	
Let $n \in \mathbb{N}$ and let $\psi,f \in \mathcal{C}^{n}([a,b],\mathbb{R})$ 
be two functions such that $\psi$ is increasing and $\psi'(t)\neq 0$ 
for all $t \in [a,b]$. The $\psi$-Riemann--Liouville fractional derivative 
of $f$ of order $\alpha$ is given by
\begin{equation}
\label{FD-RL}
\begin{split}
^{RL}D^{\alpha,\psi}_{a^+} f(t)
&=  \left(\frac{1}{\psi'(t)}  
\frac{d}{dt} \right)^n I^{n-\alpha,\psi}_{a^+} f(t) \\
&= \frac{1}{\Gamma(n-\alpha)}\left(\frac{1}{\psi'(t)}  
\frac{d}{dt} \right)^n\int_{a}^{t}\psi'(s)(\psi(t)-\psi(s))^{n-\alpha-1}f(s) ds,
\end{split}
\end{equation}
where $n=[\alpha]+1$.  
\end{definition}

\begin{definition}[The $\psi$-Caputo fractional derivative \cite{Almeida1_et_al}]
\label{def:03}	
Let $n \in \mathbb{N}$ and let $\psi,f \in \mathcal{C}^{n}([a,b],\mathbb{R})$ 
be two functions such that $\psi$ is increasing and $\psi'(t)\neq 0$ for all $t \in [a,b]$. 
The $\psi$-Caputo fractional derivative of $f$ of order $\alpha$ is given by
\begin{equation}
\label{FD-C}
\begin{split}
^{C}D^{\alpha,\psi}_{a^+} f(t)
&=  I^{n-\alpha,\psi}_{a^+}  \left(\frac{1}{\psi'(t)}  
\frac{d}{dt} \right)^n f(t) \\
&=  \frac{1}{\Gamma(n-\alpha)}
\int_{a}^{t}\psi'(s)(\psi(t)-\psi(s))^{n-\alpha-1}f_{\psi}^{[n]}(s) ds,
\end{split}
\end{equation}
where $n=[\alpha]+1$ and $f_{\psi}^{[n]}(t):=\left(\frac{1}{\psi'(t)}  
\frac{d}{dt} \right)^n f(t)$.  
\end{definition}

\begin{definition}[See \cite{MR3244285}]
The Mittag--Leffler function for one and two-parameter is defined, respectively, as
\begin{equation}
\label{M-L}
E_\alpha (z)=\sum_{k=0}^{\infty} \frac{z^k}{\Gamma(\alpha k+1)}, 
\quad \alpha \in \mathbb{C}, \quad Re(\alpha)>0,
\end{equation}
\begin{equation}
\label{M-L1}
E_{\alpha,\beta} (z)=\sum_{k=0}^{\infty} \frac{z^k}{\Gamma(\alpha k+\beta)},
\quad \alpha, \beta \in \mathbb{C}, \quad Re(\alpha)>0. 
\end{equation}
\end{definition}

We make use of generalized Laplace transforms.

\begin{definition}[See \cite{MR4577643}]
Let $f : [0,\infty) \rightarrow \mathbb{R}$ be a real-valued function and
$\psi$ be a nonnegative increasing function such that $\psi(0)=0$. 
Then the Laplace transform of $f$ with respect to $\psi$ is defined by
$$
\mathcal{L}_\psi\{f(t)\}
=F(s)
=\int_{0}^{\infty} e^{-s \psi(t)}\psi'(t) f(t) dt
$$
for all $s \in \mathbb{C}$ such that this integral converges. 
Here $\mathcal{L}_\psi$ denotes the Laplace transform with respect to $\psi$, 
which is called a generalized Laplace transform.
\end{definition}

\begin{definition}[See \cite{MR4577643}]
An $n$-dimensional function $f : [0,\infty) \rightarrow \mathbb{R}$
is said to be of $\psi$-exponential order $c>0$
if there exist positive constants $M$ and $T$ 
such that for all $t>T$, 
$$
\left\|f\right\|_\infty = \max_{1 \leq i \leq n} \left\|f_i\right\|_\infty
\leq M e^{c \psi(t)},
$$
that is, if $f(t) = O\left(e^{c \psi(t)}\right)$ as $t \rightarrow \infty$.
\end{definition}

\begin{lemma}[See \cite{Jarad_et_al}]
\label{lemma:01}
Let $\alpha >0$, $f \in AC_{\psi}^{n}[a,b]$ for any $b>a$, 
and $f^{[k]}$, $k=0,1,\ldots n$, be of $\psi(t)$-exponential order. Then,
\begin{equation}
\label{FLT}
\mathcal{L}_\psi  \{(^{C}D^{\alpha,\psi}_{a^+}f)(t)\}(s)
=s^\alpha  \left[  \mathcal{L}_\psi  \{f(t)\}- \sum_{k=0}^{n-1} s^{-k-1}(f^{[k]})(a^+)\right].
\end{equation}
\end{lemma}

\begin{lemma}[See \cite{Jarad_et_al}]
\label{lemma:02}	
Let $Re(\alpha)>0$ and $\vert \frac{\lambda}{s^\alpha}\vert <1$. Then,
\begin{equation} 
\label{GL-M1}
\mathcal{L}_\psi  \{E_\alpha \left( \lambda (\psi(t)-\psi(a))^\alpha \right) \} 
=\frac{s^{\alpha-1}}{s^\alpha-\lambda},
\end{equation}
and 
\begin{equation}
\label{GL-M2}
\mathcal{L}_\psi  \{(\psi(t)-\psi(a))^{\beta-1} 
E_{\alpha,\beta} \left( \lambda (\psi(t)-\psi(a))^\alpha \right) \} 
=\frac{s^{\alpha-\beta}}{s^\alpha-\lambda}.
\end{equation}
\end{lemma}


\section{Main results}
\label{sec:03}

We propose a new blood alcohol model associated with 
the $\psi$-Caputo fractional operator (Section~\ref{sec:03:1}),
we obtain its analytical solution (Theorem~\ref{thm:new}),
and show the advantages of our model with respect
to the ones available in the literature (Section~\ref{sec:04}). 


\subsection{Blood alcohol model and its analytical solution}
\label{sec:03:1}

Alcoholic beverages have been an integral part of various cultures for thousands of years. 
However, it is important to recognize that alcohol consumption not only leads to disorders 
but also significantly impacts the incidence of chronic diseases, injuries, and health issues. 
Understanding the effects of alcohol consumption on health requires considering three key factors: 
the quality of the alcohol consumed, the volume of alcohol consumed, and the consumption pattern.
Examining these categories is crucial as they contribute to both detrimental and beneficial consequences. 
For instance, epidemiological studies and research conducted on animals have suggested that excessive 
alcohol consumption can depress cardiac function and cause cardiomyopathy or cardiomyopathy-related injuries. 
The negative effects of heavy alcohol consumption extend however beyond cardiac function, leading to symptoms 
such as thirst, fatigue, drowsiness, weakness, nausea, dry mouth, headaches, and difficulties with concentration
\cite{Global,Keefe_et_al,Larsson_et_al,Friedman,Penning_et_al}.

By comprehending the impact of alcohol consumption in terms of quality, volume, and consumption pattern, 
we can better understand its effects on health and make informed decisions regarding alcohol consumption.
Here we propose the following model:
\begin{equation}
\label{PL}
\begin{cases}
^{C}D^{\alpha,\psi} A(t) =-k_1A(t), \quad A(0)=A_0,\\ 
^{C}D^{\beta,\psi} B(t) =k_1A(t)-k_2B(t), \quad B(0)=B_0,
\end{cases} 
\end{equation}
where $A(t)$ and $B(t)$ represent the absorptions of alcohol in stomach and alcohol 
in the blood at time $t$, respectively, $k_1$ and $k_2$
are nonzero constants, $A_0$ is the initial absorption of alcohol in the stomach,
and $B_0$ the initial quantity of alchool in the blood. We obtain the
solution for the concentration of alcohol in stomach, $A(t)$, and the concentration 
of alcohol in the blood, $B(t)$, by employing the generalized Laplace transform technique.

\begin{remark} 
\label{rm9}
Given the so called conjugation relations, see equation (18.26) of \cite{MR1347689}
or equations (2.5.7) to (2.5.10) of \cite{MR2218073}, one can easily reduce $\psi$-Caputo fractional 
differential equations to classical Caputo fractional differential equations \cite{MR4379926}.
\end{remark}

\begin{remark} 
\label{rm10}
The main differences between the $\psi$-Riemann-Liouville 
and $\psi$-Caputo fractional derivatives lie in their definitions 
(cf. Definitions~\ref{def:02} and \ref{def:03}), memory effects, 
and locality properties. The choice between these two operators 
depends on the specific mathematical requirements and physical 
interpretations of the problem at hand. In our case, they allow 
us to consider the initial conditions $A(0)=A_0$ and $B(0)=B_0$ 
in \eqref{PL}, which is in agreement with the previous models 
considered in the literature: see \cite{Ricardo_et_al,Qureshi_et_al}
and references therein.
\end{remark}

\begin{theorem}
\label{thm:new}
The solution to the system of fractional differential equations 
\eqref{PL} is given by
\begin{equation}
\label{Eq1}
A(t)=A_0  E_{\alpha} \left( -k_1 (\psi(t)-\psi(0))^\alpha \right)  
\end{equation}
and 
\begin{equation}
\label{Eq2}
\begin{split}
B(t) &=  B_0  E_{\beta} \left( -k_2 (\psi(t)-\psi(0))^\beta \right) 
+ k_1 A_0 \sum_{n=0}^{\infty}\sum_{m=0}^{\infty}
\frac{(-k_2)^n(-k_1)^m}{\Gamma((n+1)\beta)\Gamma(\alpha m+1)} \\
& \quad \times \int_{0}^{t} (\psi(t)-\psi(\tau))^{\beta+n \beta-1} 
(\psi(\tau)-\psi(0))^{\alpha m} \psi'(\tau) d\tau.   
\end{split}
\end{equation}
\end{theorem}

\begin{proof}
Taking the Laplace transform on the first equation of (\ref{PL}), 
we get from Lemma~\ref{lemma:01} that
\begin{equation*}
\begin{split}
\mathcal{L}_\psi  \{^{C}D^{\alpha,\psi} A(t) \} 
&= \mathcal{L}_\psi  \{-k_1A(t)\},\\
s^\alpha \left[ \mathcal{L}_\psi  \{ A(t) \} 
- \sum_{k=0}^{n-1} s^{-k-1}(f^{[k]})(0)\right] 
&= -k_1 \mathcal{L}_\psi  \{ A(t) \}, \\
(s^\alpha+k_1)\mathcal{L}_\psi  \{ A(t) \} 
&= s^\alpha  \sum_{k=0}^{n-1} s^{-k-1}(f^{[k]})(0),\\
\mathcal{L}_\psi  \{ A(t) \} 
&= A_0 \frac{s^{\alpha-1}}{s^\alpha+k_1},
\end{split}
\end{equation*}
and equation (\ref{Eq1}) follows by taking the inverse Laplace transform.
Using (\ref{Eq1}), the second fractional order equation becomes
\begin{equation*}
^{C}D^{\beta,\psi} B(t) =k_1A_0  E_{\alpha} 
\left( -k_1 (\psi(t)-\psi(0))^\alpha \right)-k_2B(t)
\end{equation*}
and, taking the Laplace transform, we get from Lemma~\ref{lemma:02} that
\begin{equation*}
\begin{split}
\mathcal{L}_\psi  \{^{C}D^{\alpha,\psi} B(t) \} 
&= \mathcal{L}_\psi  \{k_1A_0  E_{\alpha} \left( -k_1 (\psi(t)-\psi(0))^\alpha \right)-k_2B(t)\},\\
s^\alpha \left[ \mathcal{L}_\psi  \{ B(t) \} - \sum_{k=0}^{n-1} s^{-k-1}(f^{[k]})(0)\right] 
&= \mathcal{L}_\psi  \{k_1A_0  E_{\alpha} \left( -k_1 (\psi(t)-\psi(0))^\alpha \right)\}
-k_2 \mathcal{L}_\psi  \{ B(t) \}, \\
(s^\alpha+k_2)\mathcal{L}_\psi  \{ B(t) \} -s^{\alpha-1} B_0 
&= \mathcal{L}_\psi  \{k_1A_0  E_{\alpha} \left( -k_1 (\psi(t)-\psi(0))^\alpha \right)\},  \\
\mathcal{L}_\psi  \{ B(t) \} 
&= B_0 \frac{s^{\alpha-1}}{s^\alpha+k_2}+\frac{1}{s^\alpha+k_2}\mathcal{L}_\psi 
\{k_1A_0  E_{\alpha} \left( -k_1 (\psi(t)-\psi(0))^\alpha \right)\},\\ 
\mathcal{L}_\psi \{ B(t) \} 
&= B_0 \mathcal{L}_\psi \{ E_{\alpha} \left( -k_2 (\psi(t)-\psi(0))^\alpha \right)\} \\
& \quad + \mathcal{L}_\psi  \{ (\psi(t)-\psi(0))^{\alpha-1} 
E_{\alpha,\alpha} \left( -k_2 (\psi(t)-\psi(0))^\alpha \right)\} \\
& \quad \times \mathcal{L}_\psi \{k_1 A_0  E_{\alpha} \left( -k_1 (\psi(t)-\psi(0))^\alpha \right) \}.
\end{split}
\end{equation*}
Therefore,
\begin{multline*}
\mathcal{L}_\psi \{ B(t) \} = B_ 0 \mathcal{L}_\psi \{ E_{\alpha} \left( -k_2 (\psi(t)-\psi(0))^\alpha \right)\\
+(\psi(t)-\psi(0))^{\alpha-1} E_{\alpha,\alpha} \left( -k_2 (\psi(t)-\psi(0))^\alpha \right)*_\psi A(t) \}
\end{multline*}
and, taking the inverse Laplace transform, we get
\begin{equation}
\label{eq:B:alternat}
B(t) = B_ 0  E_{\alpha} \left( -k_2 (\psi(t)-\psi(0))^\alpha \right)
+(\psi(t)-\psi(0))^{\alpha-1} E_{\alpha,\alpha} \left( -k_2 (\psi(t)-\psi(0))^\alpha \right)*_\psi A(t), 
\end{equation}
which proves the intended expression (\ref{Eq2}).
\end{proof}

\begin{remark}
\label{rm:ref01}
Note that the proof of Theorem~\ref{thm:new} 
shows that the double series appearing in equation \eqref{Eq2} 
can be expressed in terms of the Mittag--Leffler function for two-parameters:
see \eqref{eq:B:alternat}.
\end{remark}


\subsection{Application}
\label{sec:04}

Now an application is provided to support our theoretical model \eqref{PL}. 
For that we use Blood Alcohol Levels (BAL) data of a real individual, 
using our fractional model and showing the important 
role of fractional differentiation with respect to another function $\psi$. 

Given real BAL data along time,
consisting of $r$ points, $(t_0, B_0 ), \ldots, (t_r, B_r)$, 
we approximate these values by the solution $t \mapsto B(t)$
of our theoretical model. The form $B$ is known, 
being given by \eqref{Eq2}, but it depends on $\psi$, and $\alpha$ and $\beta$. 
For each approximation $B(t_i)$ of $B_i$ given by the model, the error 
is defined as the difference between the exact and the approximated
values, that is, by $d_i :=  B_i-B(t_i)$, 
$i=1,\ldots,r$, while the total square error is given by
\begin{equation}
\label{av:total:re}
Error=\sum_{i=1}^{r} (d_i)^2.
\end{equation}


\subsubsection{The classical integer order model}
\label{sec:3.2.1}

In the particular case when we chose in our model \eqref{PL}
$\psi(t) = t$ and $\alpha = \beta = 1$, we obtain the classical 
system of ordinary differential equations that model 
the blood alcohol level \cite{Qureshi_et_al,Ludwin,Almeida_et_al}:
\begin{eqnarray}
\label{PL-int}
\begin{cases}
 \frac{dA(t)}{dt} =-k_1A(t), \quad A(0)=A_0,\\ \\
\frac{dB(t)}{dt} B(t) =k_1A(t)-k_2B(t), \quad B(0)=0.
\end{cases} 
\end{eqnarray}
The solution of problem (\ref{PL-int}) is given
as a direct corollary of our Theorem~\ref{thm:new} as
\begin{equation}
\label{Eq1:b}
A(t)=A_0 e^{-k_1 t},
\end{equation}
and
\begin{equation}
\label{Eq02:b}
B(t)=A_0 \frac{k_1}{k_2-k_1}(e^{-k_1 t}-e^{-k_2 t}).
\end{equation}

In \cite{Ludwin}, Ludwin tried to fit the experimental data given 
in Table~\ref{tab1} using the classical model \eqref{PL-int}. 
\begin{table}[ht] 
\caption{Experimental data for the Blood Alcohol Level (BAL) of a real individual \cite{Ludwin}.} 
\label{tab1}
\begin{center}
\begin{tabular}{|l|l|l|l|l|l|l|l|l|l|l|l|} \hline 
Time (min)& 0 & 10 & 20 & 30 & 45 & 80 & 90 & 110 & 170 \\ \hline
BAL (mg/L) & 0 & 150 & 200 & 160 & 130 & 70 & 60 & 40 & 20 \\ \hline 
\end{tabular}
\end{center}
\end{table}

Using an Excel solver, Ludwin found the values
\begin{equation}
\label{eq:paramter:values:CM}
A_0\approx 245.8769, \ \ \  k_1\approx 0.109456, \ \ \  k_2\approx 0.017727, 
\end{equation}
for the expression \eqref{Eq02:b} of $B(t)$. 
The modeling results from equation (\ref{Eq02:b}) 
with the parameter values \eqref{eq:paramter:values:CM}
are given in Table~\ref{tab2}. 
\begin{table}[ht] 
\caption{Blood alcohol level (BAL) predicted by the classical theoretical 
model \eqref{PL-int} with the parameter values \eqref{eq:paramter:values:CM},
corresponding to an error \eqref{av:total:re} of $775$ $(mg/L)^2$.} \label{tab2}
\begin{center}
\begin{tabular}{|l|l|l|l|l|l|l|l|l|l|l|l|} \hline 
Time (min) & 0 & 10 & 20 & 30 & 45 & 80 & 90 & 110 & 170 \\ \hline
BAL (mg/L) & 0.00 & 147.54 & 172.95 & 161.38 & 129.99 & 70.99 & 59.49 & 41.74 & 14.41 \\ \hline
\end{tabular}
\end{center}
\end{table}
The error \eqref{av:total:re} between the real data of Table~\ref{tab1}
and the values of Table~\ref{tab2} obtained by the classical model \eqref{PL-int} 
is $775$ $(mg/L)^2$. However, as shown in \cite{Qureshi_et_al}, 
these results can be improved by choosing
\begin{equation}
\label{eq:paramter:values:CM:b}
A_0 = 261.721,
\quad k_1 = 0.111946,
\quad k_2 = 0.0186294,
\end{equation}
for which the model \eqref{PL-int} gives the values of Table~\ref{tab2:b},
decreasing the error from $775$ $(mg/L)^2$ to
\begin{equation}
\label{eq:best:classical:error}
E_{classical} \approx 496 \ (mg/L)^2.   
\end{equation}
\begin{table}[ht] 
\caption{Blood alcohol level (BAL) predicted by the classical theoretical 
model \eqref{PL-int} with the parameter values \eqref{eq:paramter:values:CM:b},
corresponding to an error \eqref{av:total:re} of $496$ $(mg/L)^2$.} \label{tab2:b}
\begin{center}
\begin{tabular}{|l|l|l|l|l|l|l|l|l|l|l|l|} \hline 
Time (min) & 0 & 10 & 20 & 30 & 45 & 80 & 90 & 110 & 170 \\ \hline
BAL (mg/L) & 0.00 & 158.11 & 182.85 
& 168.62 & 133.73 & 70.69 & 58.69 & 40.44 & 13.22 \\ \hline
\end{tabular}
\end{center}
\end{table}

Such result can be improved using our fractional model \eqref{PL}.
Indeed, by other choices of function $\psi$ and $\alpha$ and $\beta$ 
in \eqref{PL}, the solution of our fractional model \eqref{PL} 
can be closer to the real data of Table~\ref{tab1}
than the solution obtained by the classical model \eqref{PL-int}.
To measure that, we follow \cite{Rosales_et_al} and define the 
gain $\mathcal{G}$ of the efficiency of our model,
comparing the error \eqref{eq:best:classical:error} of the classical model, 
$E_{classical}$, with the error \eqref{av:total:re} associated 
to a particular fractional instance of our model \eqref{PL}:
\begin{equation}
\label{eq:gain}
\mathcal{G} 
= \left\vert \frac{E_{classical}-E_{fractional}}{E_{classical}}\right\vert
\approx \left\vert \frac{496-E_{fractional}}{496}\right\vert.
\end{equation}
In percentage, we multiply the value \eqref{eq:gain} by 100.


\subsubsection{The Caputo fractional order model}
\label{sec:3.2.2}

In the particular case when we chose in our model \eqref{PL}
$\psi(t) = t$ with $\alpha, \beta \in (0,1)$, we obtain the  
standard Caputo system of fractional differential equations 
studied in \cite{Ricardo_et_al}: 
\begin{equation}
\label{PL-fra}
\begin{cases}
^{C}D^{\alpha} A(t) =-k_1A(t), \quad A(0)=A_0,\\
^{C}D^{\beta} B(t) =k_1A(t)-k_2B(t), \quad B(0)=0.
\end{cases} 
\end{equation}
The solution of (\ref{PL-fra}) is also a direct
consequence of our Theorem~\ref{thm:new}, which gives
\begin{equation*}
A(t)=A_0 E_{\alpha}(-k_1 t^\alpha),
\end{equation*}
and
\begin{equation*}
B(t)=A_0 \sum_{n=0}^{\infty}\sum_{m=0}^{\infty}
\frac{(-k_2)^n(-k_1)^m}{\Gamma(n \beta 
+ \beta +m \alpha +1)}t^{n \beta +\beta +m \alpha}.
\end{equation*}
Using the real experimental data of blood alcohol level of Table~\ref{tab1}, 
the authors of \cite{Ricardo_et_al} used a numerical optimization approach based 
on the least squares approximation to determine the orders $\alpha$ and $\beta$ 
of the fractional Caputo operator that better describes the real data. 
They proved that the Caputo fractional model \eqref{PL-fra} fits better the available data
when compared with the classical one given by \eqref{PL-int}. Moreover, in 2019, 
Qureshi et al. \cite{Qureshi_et_al} considered not only the Caputo
fractional operator, which has a singular kernel, but also non-singular kernels: 
they investigated the use of the Atangana--Baleanu--Caputo (ABC) 
and the Caputo--Fabrizio (CF) kernels to fractionalize the classical model. 
It has been shown in \cite{Qureshi_et_al} that the fractional versions based 
on ABC and CF operators are not able to improve the accuracy of the results 
obtained by the Caputo model \eqref{PL-fra}. The current state of the art 
is thus given by the fractional model \eqref{PL-fra} with the parameters 
\begin{equation}
\label{eq:paramter:values:CM:fm}
A_0\approx 991.085, \ \ \  
k_1\approx 0.0287362, \ \ \  
k_2\approx 0.0843802, \ \ \ 
\alpha \approx 0.881521, 
\ \ \ \beta =1,
\end{equation}
found via the least squares error minimization technique \cite{Qureshi_et_al},
for which the Caputo model \eqref{PL-fra} gives the values of Table~\ref{tab2:frac}.

\begin{table}[ht] 
\caption{Blood alcohol level (BAL) predicted by the Caputo theoretical model \eqref{PL-fra} 
with the parameter values \eqref{eq:paramter:values:CM:fm},
corresponding to an error \eqref{av:total:re} of $417$ $(mg/L)^2$
and a gain \eqref{eq:gain} of $16\%$.} \label{tab2:frac}
\begin{center}
\begin{tabular}{|l|l|l|l|l|l|l|l|l|l|l|l|} \hline 
Time (min) & 0 & 10 & 20 & 30 & 45 & 80 & 90 & 110 & 170 \\ \hline
BAL (mg/L) &0.000 & 157.733 & 184.469 & 169.800 & 131.700 & 
67.879 & 57.162 & 41.883 & 20.730 \\ \hline
\end{tabular}
\end{center}
\end{table}

The values of Table~\ref{tab2:frac} lead to an error of $417$ $(mg/L)^2$,
which represents a gain of $16\%$ with respect to the best
fitting of the classical model. As we shall see, we can however
improve this state of the art by using our general 
$\psi$-Caputo model \eqref{PL}
with an appropriate function $\psi$ different from the identity.


\subsubsection{The $\psi$-Caputo fractional order model}
\label{sec:3.2.3}

We now consider an application of our fractional differential system
with $\psi$-Caputo fractional derivatives to the blood alcohol concentration 
involving the two linked absorption processes $A(t)$ and $B(t)$, 
first in the stomach and then in the blood. Precisely, we provide a 
different function $\psi(t)$ for which the solutions of the fractional 
model (\ref{PL}) models better the given real data of Table~\ref{tab1}
when compared with the ones studied in the literature.

Recall that for any choice of $\psi(t)$ one can always
reduce our $\psi$-fractional system to a classical Caputo system:
according with Remark~\ref{rm9}, $\psi$-Caputo fractional problems 
are just Caputo fractional problems. Here we show that one does not
need to use nontrivial functions $\psi$ to improve the state of the art.
Indeed, in comparison to existing classical and fractional models found in
the literature, we outperform them significantly by employing a
simple yet non-standard kernel function $\psi(t)$, reducing the error 
by more than half, resulting in an impressive gain improvement of 59\%.

Let $\psi(t)=a_1+a_2 t$, $a_1$, $a_2$ $\in \mathbb{R}$. 
It follows from \eqref{Eq2} that
\begin{equation}
\label{Eq2:b}
\begin{split}
B(t)&= k_1 A_0 \sum_{n=0}^{\infty}\sum_{m=0}^{\infty}
\frac{(-k_2)^n(-k_1)^m}{\Gamma((n+1)\beta+\alpha m+1)} 
(a_2 t)^{\alpha m+\beta(n+1)}
\end{split}
\end{equation}
or, equivalently, in terms of a naturally emerging 
bivariate Mittag--Leffler function,
$$
B(t) = k_1 A_0 (a_2 t)^{\beta} E_{\alpha,\beta,\beta+1} 
\left( -k_1(a_2t)^{\alpha}, -k_2(a_2t)^{\beta} \right),
$$
where this $E$ is the naturally emerging bivariate Mittag--Leffler 
function introduced in 2020 \cite{MR4117524}.

To obtain the best possible values for the parameters
$a_1$ and $a_2$ that define $\psi(t)$ and the best
values of $\alpha$ and $\beta$, we have used
the free and open source GNU Octave 
high-level programming language
and the \texttt{lsqcurvefit} routine of the optimization
package, which solves nonlinear data fitting problems 
in the least squares sense. Precisely, we developed
the GNU Octave code of Listing~\ref{Code:01}.
\begin{lstlisting}[escapeinside={(*}{*)}, 
caption=GNU Octave code used for Section~\ref{sec:3.2.3}.]
clear, clc; (*\label{Code:01}*)
pkg load optim
t  = [0 10 20 30 45 80 90 110 170];
BAL = [0 150 200 160 130 70 60 40 20];     
t0 = t(1);
p0= [991.085 0.0287362 0.0843802 0.881521 1 1];
Error = @(R,M) sum((R-M).^2); 
% psi-Fractional case  with psi(t)=a1+a2.*t
B = @(k,t) arrayfun(@(t) Spsi(t0,t,k),t);
B1= B(p0,t)
printf('Total cpu time: %f seconds\n', cputime-cput);
Error(BAL,B1);
cput = cputime;
p = lsqcurvefit(B,p0,t,BAL);  
%printf('Total cpu time: %f seconds\n', cputime-cput);
Bp01 = B(p,t);
% Best values from paper of S. Qureshi et al. 2019 
BAL01= [0 157.733 184.469 169.80 131.70 67.879 57.162 41.883 20.730];    
figure
plot(t,Bp01,'r')
ylabel('Blood alcohol level (mg/l)')
xlabel('Time (minutes)')
hold on
plot(t,BAL01,'--g')
hold on
plot(t,BAL,"o")
legend({'Fractional psi(x)=a1+a2t','Fractional psi(x)=x','real data'})
hold off
\end{lstlisting}

Using our Octave code of Listing~\ref{Code:01},
we obtained the parameter values given in Table~\ref{tab4}.
\begin{table}[ht]
\caption{Optimal parameter values for the $\psi$-Caputo model \eqref{PL}
with $\psi(t)=a_1+a_2 t$.}\label{tab4}	
\begin{center}
\begin{tabular}{|c|c|c|c|c|c|c|} \hline
$A_0$ & $k_1$ & $k_2$ & $\alpha$ & $\beta$ & $a_1$ & $a_2$ \\ \hline
1270.679 & 0.0217903 & 0.1330596 & 1.012392 & 1.288845 
& $\forall a_1 \in \mathbb{R}$ & 0.621767 \\ \hline
\end{tabular}
\end{center}
\end{table}

The Blood alcohol level \eqref{Eq2} associated with the values 
of Table~\ref{tab4} are given in Table~\ref{tab3},
\begin{table}[ht] 
\caption{Blood alcohol level (BAL) predicted by the new $\psi$-Caputo 
model \eqref{PL} with the parameter values of Table~\ref{tab4},
corresponding to an error \eqref{av:total:re} of less than $202$ $(mg/L)^2$
and a gain \eqref{eq:gain} of more than $59\%$.} \label{tab3}
\begin{center}
\begin{tabular}{|l|l|l|l|l|l|l|l|l|l|l|l|} \hline 
Time (min) & 0 & 10 & 20 & 30 & 45 & 80 & 90 & 110 & 170 \\ \hline
BAL (mg/L) & 0.000 & 152.569& 193.826&  169.420& 123.541 
& 70.156 & 60.340 & 44.502 &  17.518 \\ \hline
\end{tabular}
\end{center}
\end{table}
for which we decrease the error of $417$ $(mg/L)^2$ for the best model
in the literature to a total error of less than $202$ $(mg/L)^2$.
This corresponds to a gain of more than $59\%$ with respect
to the classical model. 

In Figure~\ref{fig2}, we plot the real data
of Table~\ref{tab1} with the curves obtained
with $\psi(t) = t$ (Section~\ref{sec:3.2.2}) 
and $\psi(t)=a_1+a_2 t$ (Section~\ref{sec:3.2.3}).
\begin{figure}
\centering
\includegraphics[scale=0.6]{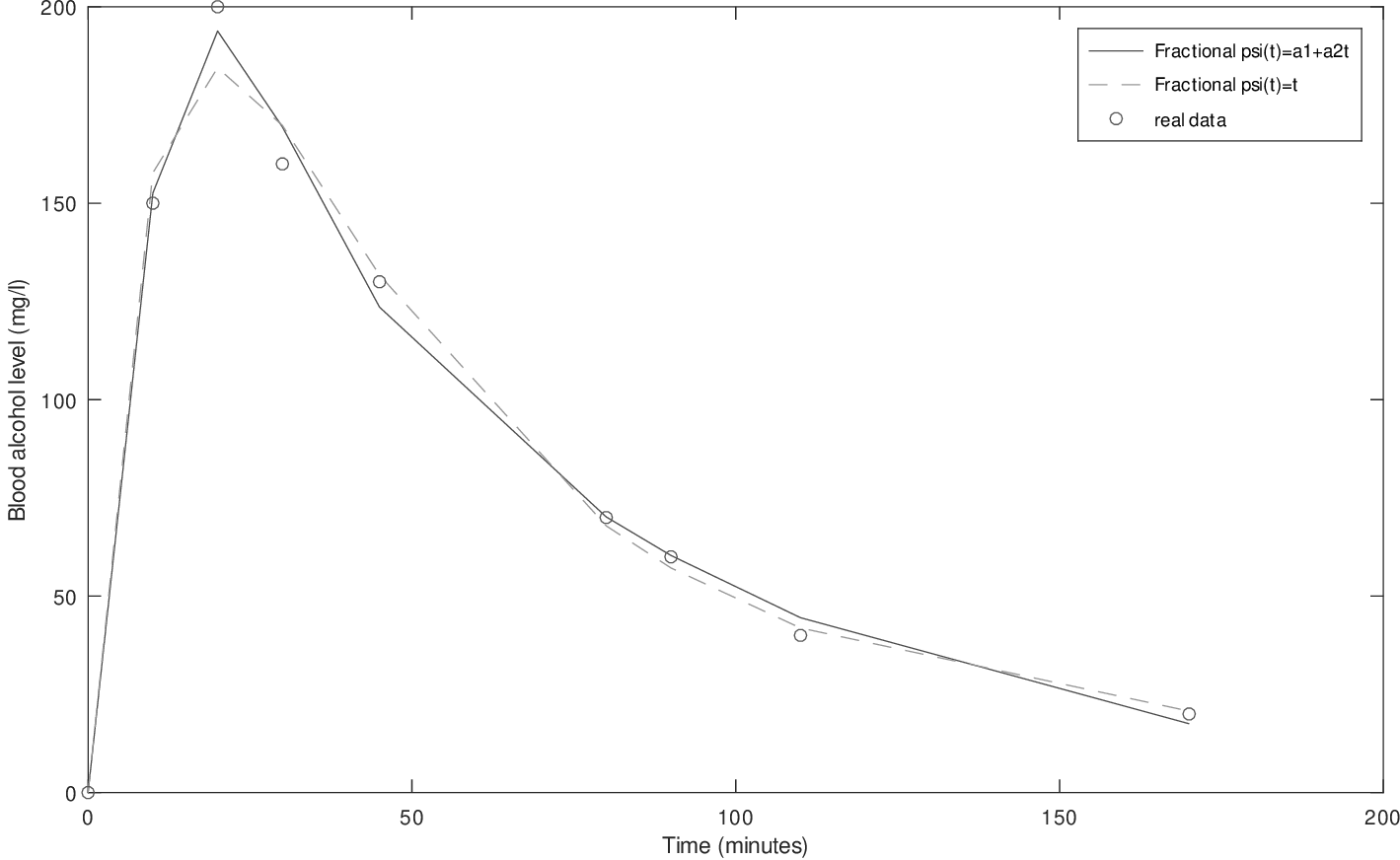}
\caption{Blood alcohol level comparison between the real data 
of Table~\ref{tab1} and the predictions obtained from
the best fractional models \eqref{PL} with $\psi(t) = t$ (Caputo)
and $\psi(t)=0.621767 t$.}
\label{fig2}
\end{figure}

If ones only uses $\psi(t) = a_1 + a_2 t$, then it is not 
really $\psi$-fractional calculus, but only a constant multiple 
of classical fractional calculus. Indeed, putting $\psi(t) = a_1 + a_2 t$ 
in equation \eqref{IF-RL}, it is clear that the $\psi$-Riemann--Liouville 
fractional integral to order alpha is simply $a_2^\alpha$ times 
the Riemann--Liouville fractional integral to order $\alpha$. 
Similarly, from equations \eqref{FD-RL}--\eqref{FD-C}, it is clear that the 
$\psi$-fractional derivative to order $\alpha$ (Riemann--Liouville or Caputo) 
is simply $a_2^\alpha$ times the fractional derivative to order $\alpha$ 
(Riemann--Liouville or Caputo) when $\psi(t) = a_1 + a_2 t$. 
Therefore, in this case, all of our $\psi$-Caputo fractional models 
are just Caputo fractional models -- constant multiples do not change 
the shape of the problem. To demonstrate the usefulness of $\psi$-Caputo fractional calculus, 
we end with an example where we use an actual $\psi$-Caputo model, 
comparing it with the models available in the literature. 
As one can see from Table~\ref{tab7}, the choice $\psi(t)=(t+0.5)^{0.97}$
is enough to improve the results published in the literature.

\begin{table}[ht] 
\caption{Comparison between integer-order, Caputo and $\psi$-Caputo models.}
\label{tab7}
\begin{center}
\begin{tabular}{|c|c|c|c|c|c|c|} \hline 
$\psi(t)$ & $\alpha$ & $\beta$ & $A_0$ & $k_1$ & $k_2$ & Error \\ \hline
integer order &1 & 1 & 245.8769 & 0.109456 & 0.017727& 775.2226 \\ \hline
$\psi(t)=t$ (Caputo) &0.979 & 0.979 & 991.085 & $0.030960$ &$0.08887702$ & 1383,82052 \\ \hline
$\psi(t)=(t+0.5)^{0.97}$ &0.972 & 0.972 & $850.085$ & $0.030960$ &$0.08887702$ & 756.9497 \\ \hline
\end{tabular}
\end{center}
\end{table}


\section{Conclusion} 
\label{sec:05}

We have introduced a novel blood alcohol concentration model 
that captures the dynamics using a fractional differential 
equation featuring the $\psi$-Caputo fractional derivative. 
The utilization of the $\psi$-Caputo operator ensures an optimal 
curve fitting by allowing for the selection of a specific kernel 
$\psi$ based on the particular data being studied. By considering 
$\psi$ as a first-degree polynomial, our results demonstrate 
significant improvement compared to existing literature.
Specifically, the total square error, as shown in 
equation \eqref{av:total:re}, is reduced from $496$ $(mg/L)^2$ 
using the classical model with ordinary differential equations 
to $417$ $(mg/L)^2$ with the Caputo fractional model. However, 
with our $\psi$-Caputo model, employing $\psi(t)=0.621767 t$, 
we achieve a remarkable reduction in the total error to just 
$202$ $(mg/L)^2$, resulting in a substantial gain of $59\%$.

In summary, the key points of advantages 
of our research are:
\begin{itemize}
\item We provide a novel dynamical model for blood alcohol concentration; 

\item We successfully derive an analytic solution for both the alcohol 
concentration in the stomach and the alcohol concentration in the blood of an individual;

\item We prove analytical formulas that provide a straightforward numerical scheme;

\item We improved the state of the art with a better fit to real experimental 
data on blood alcohol levels;

\item Our model outperforms available ones significantly: 
we are able to reduce the error by more than half, resulting 
in an impressive gain improvement of 59 percent.
\end{itemize}

Given the good results obtained, for future work we plan 
to investigate the usefulness of generalized $\psi$-Caputo operators 
in other contexts, e.g. with respect to the respiratory syncytial 
virus infection \cite{MyID:419,MyID:534}.


\section*{Acknowledgements}

The authors are supported by the Portuguese Foundation for Science and Technology (FCT)
through the Center for Research and Development in Mathematics and Applications (CIDMA),
grants UIDB/04106/2020 and UIDP/04106/2020, and within the project ``Mathematical Modelling 
of Multi-scale Control Systems: Applications to Human Diseases'' (CoSysM3), 
reference 2022.03091.PTDC, financially supported by national funds (OE) 
through FCT/MCTES. They are grateful to three anonymous referees 
for several suggestions and comments that helped them to improve the paper.


\section*{Conflict of Interest Statement}

The authors have no conflicts of interest to disclose.


\section*{Data Availability Statement}

Data sharing is not applicable to this article as no new data were created.



\end{document}